\documentclass[11pt, letterpaper]{article}

\usepackage{fullpage}
\usepackage{amsthm}
\usepackage{amsmath,amssymb,amsfonts,nicefrac}
\usepackage{xspace}
\usepackage{color}
\usepackage{url}
\usepackage{bm}
\usepackage{bbm}
\usepackage{times}
\newtheorem{theorem}{Theorem}[section]
\newtheorem*{thm2}{Theorem}

\newtheorem{definition}[theorem]{Definition}

\newtheorem{lemma}[theorem]{Lemma}

\newcommand{\junk}[1]{}
\newcommand{\ignore}[1]{}
\newcommand{\E}{\mathbb{E}}

\newcommand{\eps}{\varepsilon}

\newcommand{\ol}[1]{\ensuremath{\overline{#1}}\xspace}
\newcommand{\wh}[1]{\ensuremath{\widehat{#1}}\xspace}
\newcommand{\mc}[1]{\ensuremath{\mathcal{#1}}\xspace}
\newcommand{\tn}[1]{\ensuremath{\textnormal{#1}}\xspace}
\newcommand{\poly}{\ensuremath{\textnormal{poly}}\xspace}

\newcommand{\LC}{{\sc LabelCover}\xspace}

\newcommand{\initOneLiners}{%
    \setlength{\itemsep}{0pt}
    \setlength{\parsep }{0pt}
    \setlength{\topsep }{0pt}
}

\def\showauthornotes{1}
\ifnum\showauthornotes=1
\newcommand{\Authornote}[2]{{\sf\small\color{red}{[#1: #2]}}}
\else
\newcommand{\Authornote}[2]{}
\fi

\begin{document}
\title{Hardness of Finding Independent Sets in $2$-Colorable \\
Hypergraphs and of Satisfiable CSPs}
\author{Rishi Saket\thanks{IBM India Research Lab, Bangalore, India.
Email : {\tt rissaket@in.ibm.com}. Work done in part as a
Post-Doctoral Researcher at IBM T.J. Watson Research Center, NY, USA.}}
\maketitle
\thispagestyle{empty}
\setcounter{page}{0}
\begin{abstract}
This work revisits the PCP Verifiers used in the works of H\aa
stad~\cite{Hastad}, Guruswami et al.~\cite{GHS},
Holmerin~\cite{Holmerin} and Guruswami~\cite{guruswami} 
for \emph{satisfiable}
{\sc Max-E$3$-SAT} and {\sc Max-E$k$-Set-Splitting}, and independent
set in $2$-colorable $4$-uniform
hypergraphs. We provide simpler and
more efficient PCP Verifiers to prove the following improved hardness
results: 

\smallskip
\noindent
Assuming that NP $\not\subseteq$ DTIME$(N^{O(\log\log N)})$,
\begin{itemize}
\item There is no polynomial time algorithm that, given an $n$-vertex
$2$-colorable $4$-uniform hypergraph, finds an independent set of
$\frac{n}{(\log n)^c}$ vertices, for some constant $c >
0$.

\item There is no polynomial time algorithm that satisfies
$\frac{7}{8} + \frac{1}{(\log n)^c}$ fraction of the clauses of a 
\emph{satisfiable} {\sc Max-E$3$-SAT} instance of size $n$, 
for some constant $c >
0$.

\item For any fixed $k \geq 4$, there is no polynomial time 
algorithm that finds a partition splitting 
$(1 - 2^{-k+1}) + \frac{1}{(\log n)^c}$ fraction of the
$k$-sets of a \emph{satisfiable} {\sc Max-E$k$-Set-Splitting}
instance of size $n$, for some constant $c > 0$.
\end{itemize}
Our hardness factor for independent set in $2$-colorable $4$-uniform
hypergraphs is an exponential improvement over the previous results
of Guruswami et
al.~\cite{GHS} and Holmerin~\cite{Holmerin}. 
Similarly, our inapproximability of $(\log n)^{-c}$ beyond the 
random assignment threshold for {\sc Max-E$3$-SAT} and 
{\sc Max-E$k$-Set-Splitting} is an exponential improvement over the
previous bounds proved in
\cite{Hastad}, \cite{Holmerin} and \cite{guruswami}.

The PCP Verifiers used in our results avoid the use of a
\emph{variable bias} parameter used in previous works, which leads to
the improved hardness thresholds in addition to simplifying the
analysis substantially. Apart from standard techniques from
Fourier Analysis, for the first mentioned result we use a mixing 
estimate of Markov Chains based
on \emph{uniform} reverse hypercontractivity over general product spaces
from the work of Mossel et al.~\cite{MOS}.

\end{abstract}

\newpage

\section{Introduction}
A $k$-uniform hypergraph consists of a set of \emph{vertices} and a
collection of \emph{hyperedges} where each hyperedge is a subset of
exactly $k$ vertices. A hypergraph is said to be \emph{$q$-colorable} if its
vertices can be colored with $q$ distinct colors such that no
hyperedge contains all vertices of the same color. A related notion is
that of an \emph{independent set}, which is a subset of vertices that
does not completely contain any hyperedge.  It is easy to see
that a $q$-colorable hypergraph has at least one independent set of
$q^{-1}$ fraction of vertices, i.e. \emph{relative} size. 

Computing the
minimum number $q$ -- the \emph{chromatic number}
-- of colors required to color a hypergraph 
is a very well studied optimization problem. 
There is a simple polynomial time algorithm to decide
whether a given graph ($k = 2$) can be colored using 
$q = 2$ colors, i.e. is bipartite. 
However, for $k \geq 3$ or $q \geq 3$, 
this problem is NP-hard. A
natural question in this context is how well can the chromatic number
be approximated.

The first strong inapproximability for hypergraph coloring was given
by Guruswami, H\aa stad and Sudan \cite{GHS} who showed that it is
NP-hard to color an $n$-vertex $2$-colorable $4$-uniform hypergraph 
using constantly many
colors, and quasi-NP-hard\footnote{For ease of presentation,
in this paper we exclusively use a \emph{stronger} notion of
\emph{quasi-NP-hardness}, i.e. a problem is quasi-NP-hard if it
admits a DTIME$(N^{O(\log \log
N)})$ reduction from {\sc $3$SAT}. This differs from the weaker
requirement of  DTIME$(N^{\tn{poly}(\log
N)})$ reductions.} to color it with 
$O\left((\log\log\log n)^{-1}\log\log
n\right)$ colors. They used a notion of 
\emph{covering complexity} combined with
techniques developed in the seminal work of H\aa stad~\cite{Hastad}. 
In particular, the Probabilistically Checkable Proof (PCP)
verifier of \cite{GHS} is identical to the one used in
\cite{Hastad} for the satisfiable {\sc Max-E$4$-Set-Splitting} problem. 
Subsequently, Holmerin~\cite{Holmerin}
used a more direct approach -- with the same PCP verifier -- 
to obtain a qualitatively stronger
result. Holmerin showed that,  given a $2$-colorable
$4$-uniform it is NP-hard to compute an independent set of relative
size $\delta$,
for any constant $\delta > 0$, and it is quasi-NP-hard to do so for
 $\delta = \Omega\left((\log\log n)^{-1}\log\log\log n\right)$. 

In this work we prove the following quantitatively stronger 
result with an
exponential improvement in the hardness factor.
\begin{theorem}\label{thm-main1}
Given an $n$-vertex $2$-colorable $4$-uniform hypergraph it is quasi-NP-hard to
find an independent set of relative size $\frac{1}{\left(\log n\right)^c}$ for
some constant $c > 0$.
\end{theorem} 
As mentioned above, the results of \cite{GHS} and \cite{Holmerin} are
based on the PCP verifier used by H\aa stad \cite{Hastad} for
satisfiable {\sc Max-E$4$-Set-Splitting}. In this problem 
the input is a ground set
and a collection of its $4$-sets, and the goal is to partition the
ground set into two subsets to maximize the number of \emph{split} 
$4$-sets.  
Another fundamental constraint satisfaction problem studied by
H\aa stad \cite{Hastad} is {\sc Max-E$3$-SAT}, where the goal is to
satisfy the maximum number of a collection of $3$-literal clauses. 
H\aa stad showed
that approximating both these problems -- on satisfiable instances -- 
within $\delta$ of their random assignment threshold of $\frac{7}{8}$ is NP-hard for any 
constant $\delta > 0$ and quasi-NP-hard for
$\delta = \Omega\left((\log \log n)^{-1}\log\log\log n\right)$. Using a
strengthened analysis of Holmerin~\cite{Holmerin},
Guruswami~\cite{guruswami} extended the inapproximability to {\sc
Max-E$k$-Set-Splitting}, for any constant $k\geq 4$, with the
corresponding threshold of $\left(1-2^{-k+1}\right)$. 

In this
work we prove the following hardness thresholds for these problems 
improving exponentially the
non-constant parameter $\delta$.
\begin{theorem}\label{thm-main2}
Given an instance of {\sc Max-E$3$-SAT} of size $n$, 
it is quasi-NP-hard to decide
whether it is satisfiable or at most $\frac{7}{8} + \delta$ fraction of the
clauses can be satisfied, where $\delta = \frac{1}{\left(\log
n\right)^c}$ for some positive constant $c > 0$.
\end{theorem}

\begin{theorem}\label{thm-main3}
For any fixed $k \geq 4$, given an instance of {\sc Max-E$k$-Set-Splitting} of size $n$, 
it is quasi-NP-hard to decide
whether there is a partition of the ground set into two subsets 
splitting all the $k$-sets in the collection or at most 
$\left(1 - 2^{-k+1}\right) + \delta$ fraction of the $k$-sets are split by any
such partition, 
where $\delta = \frac{1}{\left(\log
n\right)^c}$ for some positive constant $c > 0$.
\end{theorem}

The results of this paper 
are obtained using simpler PCP verifiers for the
above problems, compared to the ones used by
H\aa stad~\cite{Hastad}, Guruswami et al.~\cite{GHS} and
Holmerin~\cite{Holmerin}. In particular, we avoid the use of a \emph{variable
bias} parameter, which yields an exponential improvement
in the inapproximability thresholds. This also considerably simplifies
our analysis compared to previous works. 
In addition, for proving Theorem \ref{thm-main1}, we are
able to use an estimate of the mixing probability of Markov Chains
over general product spaces shown -- using \emph{uniform} reverse
hypercontractivity -- by Mossel
et al.~\cite{MOS}. The proofs of Theorems \ref{thm-main2} and
\ref{thm-main3} 
use well known techniques from
Fourier Analysis, while avoiding some of the complications in previous
results. 
We remark that the starting point of our hardness
reductions is the standard Label Cover problem instead of the
so-called \emph{Smooth} Label Cover which can also 
be used to avoid the
variable bias but incurs the same loss in the hardness
factors~\cite{KhotSLC}. Section \ref{sec-tech} elaborates more on the
techniques used in this paper.

Our results also yield similar improvements in the
 hardness for satisfiable instances for other
predicates whose inapproximability in \cite{Hastad}
is shown to follow from the PCP verifiers used for satisfiable
{\sc Max-E$3$-SAT} and {\sc Max-E$4$-Set-Splitting}. The reader is referred to
Theorems 6.15, 6.18, 7.17 and 7.18 and Section 9 of
\cite{Hastad} for more details on these predicates.

\subsection{Problem Definition}
For a hypergraph $G$, let ${\sf IS}(G)$ be the size of its maximum
independent set and let $\chi(G)$ be its chromatic number.
The following is the problem of finding independent sets
in $q$-colorable hypergraphs.

\begin{definition}
{\sc ISColor}$(k, q, Q)$ : 
Given a $k$-uniform hypergraph $G(V, E)$, decide between, 
\begin{center}
 \tn{(i) YES Case:} $\chi(G) \leq q$.  {\hskip 4em} \tn{(ii) NO Case:} ${\sf IS}(G) < \frac{|V|}{Q}$.
\end{center}
\end{definition}
The problem defined above is a generalization of hypergraph coloring:
if {\sc ISColor}$(k, q, Q)$ is NP-hard for some
parameters $q, Q \in \mathbb{Z}^+$ then it is
NP-hard to color a $q$-colorable $k$-uniform hypergraph with $Q$ colors. 

The following constraint satisfaction problems are studied in this paper. 
\begin{definition} An {\sc E$k$-CNF} formula is a conjunction of
clauses (disjunctions), where each clause has exactly $k$ literals. It
is said to be \emph{satisfiable} is there is an assignment to the
variables such that each clause has at least one true literal, i.e. is
satisfied.
\end{definition}
\begin{definition}
An instance of {\sc Max-E$k$-SAT} is an {\sc E$k$-CNF} formula, and the
goal is to find an assignment to satisfy the maximum number of
clauses. In satisfiable {\sc Max-E$k$-SAT}, the input is a satisfiable
{\sc E$k$-CNF} formula.
\end{definition}
In this paper we study the above for $k=3$, i.e. {\sc
Max-E$3$-SAT}.
\begin{definition}
An instance of {\sc Max-E$k$-Set-Splitting} is a ground set and a
collection of its subsets, each of size exactly $k$. The goal is to
find a partition of the ground set into two subsets to maximize the
number of \emph{split} $k$-sets in the collection, i.e. which are not
contained in one of the subsets of the partition. In satisfiable {\sc
Max-E$k$-Set-Splitting}, the input admits a partition that 
splits all $k$-sets in the
collection. 
\end{definition}

\subsection{Previous Work}

The problem of finding large independent
sets in $q$-colorable graphs and hypergraphs (for small values of $q$) 
is very well studied
algorithmically. On $2$-colorable, i.e. bipartite,
graphs, the maximum independent set can be computed in polynomial
time. A long line of research  --  
\cite{Wigderson}, \cite{Blum}, \cite{KMS}, \cite{BK},
\cite{ACC}, and \cite{KT12} -- 
has shown that a $3$-colorable graph can be efficiently 
colored with $n^\alpha$ colors thus solving 
{\sc ISColor}$(2, 3, n^{\alpha})$. The current best value of
$\alpha \approx 0.2038$ is due to \cite{KT12}. 
For $2$-colorable
$3$-uniform hypergraphs Krivelevich et al.~\cite{KNS} gave a coloring
algorithm using $O(n^{1/5})$ colors. An upper bound of $O(n^{3/4})$ was shown
for coloring $2$-colorable $4$-uniform hypergraphs 
by Chen and Frize~\cite{CF} and Kelsen, Mahajan and
Ramesh~\cite{KMR}.

On the complexity side, the work of Guruswami,
H\aa stad and Sudan~\cite{GHS} and Holmerin~\cite{Holmerin} showed
that {\sc ISColor}$\left(4,2, O\left((\log\log\log
n)^{-1}(\log\log n)\right)\right)$ is
quasi-NP-hard. Khot~\cite{Khot-color, Khot-3} showed the
inapproximability of 
{\sc ISColor}$(4,5, (\log n)^c)$ and  
{\sc ISColor}$(3,3, (\log \log n)^c)$.
Assuming the so called \emph{Alpha Conjecture}, Dinur
et al.~\cite{DMR} showed that {\sc ISColor}$(2,3, C)$ is NP-hard for
arbitrarily large constant $C>0$. Recently, assuming the
\emph{$d$-to-$1$ Games Conjecture}, Khot
and Saket~\cite{KS14} showed that 
{\sc ISColor}$(3,2,C)$ is similarly NP-hard.

In another recent work, Dinur and Guruswami~\cite{DG13} showed a
hardness factor of $\textnormal{exp}\left(2^{\sqrt{\log\log
n}}\right)$ for a variant of coloring $2$-colorable $6$-uniform hypergraphs. 
They also showed that {\sc ISColor}$(6,2,(\log n)^c)$ is
quasi-NP-hard. The former result is obtained via a novel use of the
recently introduced
\emph{Short} Code, while the latter result uses a more standard PCP
verifier based on the \emph{Long} Code. Building upon \cite{DG13} and
concurrent to our work,  Guruswami, Harsha, H\aa stad, Srinivasan and
Varma~\cite{GHHSV13} proved the first super-polylogarithmic hardness
for hypergraph coloring showing, in particular,
the hardness of \\ {\sc ISColor}$\left(8,2,\textnormal{exp}\left(2^{\sqrt{\log\log
n}}\right)\right)$,  
{\sc ISColor}$\left(4,4,\textnormal{exp}\left(2^{\sqrt{\log\log
n}}\right)\right)$ and
 {\sc ISColor}$\left(3,3,(\log n)^{O(1/\log\log\log n)}\right)$. 
However, previous to our work the best
inapproximability for case of $2$-colorable $4$-uniform hypergraphs
remained the result of \cite{GHS}.

For satisfiable  {\sc Max-E$3$-SAT} studied in this paper, 
the random assignment gives a
$\frac{7}{8}$ approximation. Karloff and
Zwick~\cite{KZ} showed a semi-definite programming (SDP) relaxation 
based algorithm yields the
same factor on instances where each clause has \emph{at
most} $3$
literals. Their algorithm can be used to obtain a (folklore) $\frac{7}{8} +
\delta$
approximation in time $\poly(n)2^{O(\delta n)}$. H\aa stad~\cite{Hastad} showed
the inapproximability of satisfiable {\sc Max-E$3$-SAT} beyond the random assignment
threshold. In particular an approximation of  $\frac{7}{8}
+ \delta$ is NP-hard for any constant $\delta > 0$ and quasi-NP-hard
for $\delta = \Omega\left((\log\log n)^{-1}(\log\log\log n)\right)$.
On the other hand, {Max-E$3$-Set-Splitting} is known to admit an
approximation factor of $0.912$ in the satisfiable
case, while the best inapproximability is $\frac{19}{20} + \delta$ by
Guruswami~\cite{guruswami}. However, satisfiable {\sc Max-E$4$-Set-Splitting} 
was shown by H\aa stad~\cite{Hastad} to be
hard to approximate beyond its random assignment 
threshold, i.e. an approximation
factor of $\frac{7}{8} + \delta$ is NP-hard for 
for any constant $\delta > 0$ and quasi-NP-hard
for $\delta = \Omega\left((\log\log n)^{-1}(\log\log\log n)\right)$.
Guruswami~\cite{guruswami} extended this to
satisfiable {\sc Max-E$k$-Set-Splitting} for $k\geq 5$ with a
corresponding inapproximability of $\left(1-2^{-k+1}\right)  +\delta$.
The techniques used in the above results can also be combined with the
subconstant error Label Cover of Moshkovitz and Raz~\cite{MR} to 
obtain NP-hardness for $\delta = \Omega\left((\log\log n)^{-O(1)}\right)$.

\subsection{Our Results}
This paper shows the following quasi-NP-hardness results obtained via
DTIME$\left(N^{O(\log\log N)}\right)$ reductions from {\sc $3$SAT}.
\begin{thm2}\textnormal{[Theorem \ref{thm-main1}]}
{\sc ISColor}$(4,2,(\log n)^c)$ is quasi-NP-hard for some constant $c
> 0$.
\end{thm2}
\begin{thm2}\textnormal{[Theorem \ref{thm-main2}]}
Satisfiable {\sc Max-E$3$-SAT} on $n$ variables 
is quasi-NP-hard to approximate within
$\frac{7}{8} + \frac{1}{(\log n)^c}$ for some constant $c > 0$.
\end{thm2}
\begin{thm2}\textnormal{[Theorem \ref{thm-main3}]}
For any $k \geq 4$, satisfiable 
{\sc Max-E$k$-Set-Splitting} on a ground set of $n$
elements  is quasi-NP-hard to approximate within
$\left(1 - 2^{-k+1}\right) + \frac{1}{(\log n)^c}$ for some constant $c > 0$.
\end{thm2}
The proofs of Theorems \ref{thm-main1} and \ref{thm-main2} are given
in Sections \ref{sec-hyper} and \ref{sec-max3} respectively. Theorem
\ref{thm-main3} follows from the following inapproximability of {\sc
Max-E$4$-Set-Splitting}.
\begin{theorem}\label{thm-main4}
There is a DTIME$\left(N^{O(\log\log N)}\right)$ reduction from {\sc $3$SAT}
to an instance of {\sc Max-E$4$-Set-Splitting} over a ground set of
size $n$ such that:

\smallskip
\noindent
\tn{YES Case:} There is a partition of the ground set which splits every
$4$-set of the instance.

\smallskip
\noindent
\tn{NO Case:} Any fraction $\rho > 0$ of the ground set completely 
contains at
least $\rho^4 - \frac{1}{(\log n)^c}$ fraction of the $4$-sets of the
instance, for some constant $c > 0$.
\end{theorem}
Theorem \ref{thm-main4} is proved in Section \ref{sec-4ss}. For the reduction from Theorem \ref{thm-main4} to Theorem
\ref{thm-main3} we refer the reader to Theorem 8 of \cite{guruswami}.

\subsection{Techniques}\label{sec-tech}
The results of this paper, as well as those of \cite{Hastad}, \cite{GHS} and 
\cite{Holmerin} are obtained by constructing  PCPs based on
Long Codes, i.e. the verifier accepts or rejects based on a $3$
or $4$ query test of the supposed Long Code encodings. The main
technical difference from previous works is our construction of the
these tests. Let us for now focus on PCP verifier used to
prove hardness of independent set in $2$-colorable $4$-uniform
hypergraphs, and abstract out the essence of the verifier's test.
The main ingredient is a joint distribution over 
$x, x', y, y' \in \{-1,1\}^d$ for some parameter $d$, to
satisfy (among others) the following property: for each $i, j \in
[d]$, $(x_i,x'_i, y_j,y'_j)\not\in \{(1,1,1,1),(-1,-1,-1,-1)\}$. 
The tuples $x$ and $x'$ constitute the
building blocks of
two queries made by the verifier to one purported Long Code encoding, whereas $y$ and $y'$ 
correspond to those made to a second encoding. 

The difference 
of our PCP test in
comparison to that of previous works is illustrated in the distribution
of the pair $(x,x')$ (which is identical to that of $(y,y')$). 
In previous works this is induced by
the following randomized process. Select $x$ uniformly at random 
and with probability
$\frac{1}{2}$ set $x' = -x$, otherwise: for each $i\in[d]$
independently, set 
$x'_i = x_i$ with
probability $1-\eps$ and $x'_i = -x_i$ with probability $\eps$. Here
$\eps$ is a bias parameter which itself is chosen from an
appropriate distribution. This variable bias leads to an
exponential loss in the hardness factor. 

In our case, the distribution on $(x,x')$ is simpler and obtained as 
follows.
Select $x$ u.a.r. and with probability $\frac{1}{2}$ set $x' = -x$,
otherwise select $x'$ u.a.r. This test does not use any variable bias. 
Additionally, $-x'$ is a $\frac{1}{2}$-correlated 
copy of $x$ which allows us to use results of Mossel,
Oleszkiewicz and Sen~\cite{MOS} to estimate the mixing
probability of the corresponding Markov Chain. This estimate is based on 
the uniform reverse hypercontractivity proved in the same work,
wherein 
the parameters do not depend on the measure of the smallest atom 
in the probability
space. This property -- unlike the usual hypercontractivity
inequality -- is crucial for us, as the smallest atom in our
application has measure exponential in $d$, which one cannot afford. 

The PCP verifiers for satisfiable {\sc Max-E$3$-SAT} and {\sc
Max-E$4$-Set-Splitting} also use similar distributions as above. While
their analysis does not require any mixing probability estimate, 
the avoidance of the variable bias improves the
inapproximability threshold and simplifies the analysis substantially.

\section{Preliminaries}
Let us define the following notion of $\rho$-correlated spaces used by
Mossel et al.~\cite{MOS}.
\begin{definition}\label{def-1}
Consider a product space $(\Omega, \mu) = (\prod_{i=1}^n\Omega_i,
\otimes_{i=1}^n\mu_i)$ where $(\Omega_i, \mu_i)$ are finite
probability spaces. We say that $(X, Y)\in \Omega^2$ are $\rho$-correlated 
if $X$ is distributed according to
$\mu$ and the conditional distribution of $Y$ given $X$ is as
follows: for each $i$ independently,
with probability $\rho$, $Y_i = X_i$ and with probability $1 - \rho$,
$Y_i$ is sampled independently from $\mu_i$.
\end{definition}
For our analysis in Section \ref{sec-hyper} we require an estimate of
the mixing probability of Markov Chains over general product spaces.
The corresponding bound for the case of the boolean hypercube was
proved by Mossel, O'Donnell, Regev, Steif and Sudakov~\cite{MORSS06}
using reverse hypercontractivity over the boolean domain. The
generalization below was subsequently shown
by Mossel et al.~\cite{MOS}, using uniform reverse
hypercontractivity over general product spaces proved in the same work. 
We refer the reader to
\cite{MOS} for more details.
\begin{theorem}\label{thm-MOS2}
Let $(\Omega, \mu)$ be the product probability space in Definition
\ref{def-1}. Let $A, B \subseteq \Omega$ be two sets such that
$\mu\{A\}, \mu\{B\} \geq \delta \geq 0$. Let $X$ be distributed
according to the product measure $\mu$ and $Y$ be a $\rho$-correlated
copy of $X$ for some $0 \leq \rho \leq 1$. Then,
$$ \Pr\left[X \in A, Y\in B\right] \geq \delta^{\frac{2-\sqrt{\rho}}{1 -
\sqrt{\rho}}}.$$
\end{theorem}

The starting point of the reductions in this paper is the \LC problem which is defined as follows.
\begin{definition}\label{def-LC}
An instance $\mc{L}$ of \LC consists of a bipartite graph $G(U,V, E)$
along with label sets $[k]$ and $[m]$. For each edge $e$ between $u
\in U$ and $v\in V$, there is a projection $\pi_{vu} : [m]\mapsto
[k]$. A labeling $l_u \in  [k]$ to $u$ and $l_v\in [m]$ to $v$ satisfies the
edge if $\pi_{vu}(l_v) = l_u$. The goal is to find a labeling  
of $U$ and $V$ to satisfy the
maximum number of edges.
\end{definition} 

The inapproximability of \LC stated below follows from the PCP
Theorem~\cite{AS, ALMSS}, Raz's Parallel Repetition Theorem~\cite{Raz}
and a structural property proved by H\aa stad~\cite{Hastad}.
\begin{theorem}\label{thm-LC}
For every positive integer $r$, there is a deterministic $N^{O(r)}$
time reduction from a {\sc $3$SAT} instance 
of size $N$ to an instance $\mc{L}(G(U, V, E),
\{\pi_{vu}\}_{\{v,u\}\in E}, [k], [m])$ of \LC with the following
properties:
\begin{itemize}
\item[a.] $|U|, |V| \leq N^{O(r)}$. $k, m \leq 2^{3r}$. 
$G$ is bi-regular with left and right degrees bounded by
$2^{O(r)}$. 
\item[b.] There is a universal constant $c_0 > 0$ such that for any $v
\in V$ and $S \subseteq  [m]$, taking an expectation over a random
neighbor $u$ of $v$,
$$\E\left[\left|\pi_{vu}(S)\right|^{-1}\right] \leq |S|^{-2c_0}.$$
The above implies that with probability over a random neighbor $u$ of
$v$,
$$\Pr\left[\left|\pi_{vu}(S)\right| < |S|^{c_0}\right] \leq
|S|^{-c_0}.$$
\item[c.] There is a universal constant $\gamma_0 > 0$ such that,

\smallskip
\noindent
\tn{YES Case:} If the {\sc $3$SAT} instance is satisfiable then there is a
labeling to $U$ and $V$ that satisfies all edges of $\mc{L}$.

\smallskip
\noindent
\tn{NO Case:} If the {\sc $3$SAT} instance is unsatisfiable then any labeling to
$U$ and $V$ satisfies at most $2^{-\gamma_0r}$ fraction of the edges.

\end{itemize}
\end{theorem}

\section{Independent Set in $2$-Colorable $4$-Uniform Hypergraphs}
\label{sec-hyper}
In this section we give a hardness reduction from an instance of \LC to a
$4$-uniform hypergraph proving Theorem \ref{thm-main1}.

The input is an instance $\mc{L}$  of \LC from Theorem \ref{thm-LC}
consisting of a
bipartite graph $G(U, V, E)$, label sets $[k]$ and $[m]$ and
projections $\{\pi_{vu} : [m] \mapsto [k]\ \mid\ \{u,v\}\in E, u\in U,
 v\in V\}$.  
The following is
the construction of the $4$-uniform hypergraph $\mc{G}(\mc{H}, \mc{E})$.

\medskip
\noindent
{\bf Vertices.} For each vertex $v \in V$, we have a copy  of the
binary Long Code over domain $[m]$, viz. $\mc{H}^v :=
\{-1,1\}^m$. Clearly the number of vertices in each $\mc{H}^v$ is the
same : $2^m$. The set of vertices $\mc{H}$ is the union of all the
copies, i.e. $\mc{H} =  \cup_{v \in V} \mc{H}^v$.

\medskip
\noindent
{\bf Hyperedges.} The hyperedges $\mc{E}$ are added via the following
procedure. 
\begin{itemize}
\item[1.] Choose a vertex $u \in U$ u.a.r and two of its
neighbors $v, w \in V$ independently and u.a.r. 

\item[2.] Let $x, x' \in \mc{H}^v$ and $y, y' \in \mc{H}^w$ be
chosen
as follows. For each $i \in [k]$, with probability 
$\frac{1}{2}$ do
Step 2a and with probability $\frac{1}{2}$ do Step 2b.
\begin{itemize}
\item[2a.] Independently for each $j \in \pi_{vu}^{-1}(i)$ choose 
$x_j$ u.a.r. from $\{-1,1\}$ and set $x'_j$ to be $-x_j$.
Independently for each $j \in  \pi_{wu}^{-1}(i)$ choose $y_j$ 
and $y'_j$ independently and u.a.r. from $\{-1,1\}$. 

\item[2b.] Independently for each $j \in \pi_{wu}^{-1}(i)$ choose 
$y_j$ u.a.r. from $\{-1,1\}$ and set $y'_j$ to be $-y_j$.
Independently for each $j \in  \pi_{vu}^{-1}(i)$ choose $x_j$ 
and $x'_j$ independently and u.a.r. from $\{-1,1\}$.
\end{itemize}

\item[3.] For all possible choices of $u \in U$,  $v, w \in V$, 
 $x, x' \in \mc{H}^v$ and $y, y' \in \mc{H}^w$ in the above
steps, add a hyperedge between $x, x', y, y'$.
\end{itemize}

\subsection{Analysis: YES Case}
In the YES case there is a labeling $\sigma : V \mapsto [m]$ such that
for any $u \in U$ and neighbors $v, w$ of $u$, $\pi_{vu}(\sigma(v)) =
\pi_{wu}(\sigma(w))$. We partition the vertex set $\mc{H}$ into
two disjoint subsets $\mc{H}_{-1}$ and $\mc{H}_{1}$ where,
$\mc{H}_
\ell\cap\mc{H}^v
= \{ z \in \mc{H}^v\ \mid\ z_{\sigma(v)} = \ell\}$, 
for $\ell \in \{-1,1\}$. 

Consider a choice of $u \in U$ and two of its neighbors $v$ and $w$ in
Step 1 of the hyperedges construction. Steps 2a and 2b ensure that
either $x_{\sigma(v)} = -x'_{\sigma(v)}$ or  $y_{\sigma(w)} =
-y'_{\sigma(w)}$, as $\pi_{vu}(\sigma(u)) = \pi_{wu}(\sigma(w))$.
Thus, no hyperedge lies completely in either $\mc{H}_{-1}$
or $\mc{H}_{1}$ and hypergraph $\mc{G}$ is $2$-colorable.

\subsection{Analysis: NO Case}
Suppose for a contradiction that there is an independent set $\mc{I}
\subseteq \mc{H}$ such that $|\mc{I}| \geq \delta |\mc{H}|$. Our
analysis shall show that this implies a labeling to the \LC instance
$\mc{L}$ that satisfies $\delta^{O(1)}$ fraction of its edges. 
This is in contrast to the bound of
$\delta^{O(\delta^{-1})}$ obtained in \cite{GHS}, \cite{Holmerin}.

By
averaging, for at least $\delta/2$ fraction of the vertices $v \in V$, 
$|\mc{I}\cap \mc{H}^v| \geq
(\delta/2)|\mc{H}^v|$. Call such vertices as ``good''
vertices. We use $\mc{I}^v$ to denote $\mc{I}\cap
\mc{H}^v$ for any $v\in V$.

For now fix a choice of ``good'' vertices $v$ and $w$ that share a 
neighbor $u
\in U$. Let $A : \mc{H}^v \mapsto \{0,1\}$ be the indicator of the
subset $\mc{I}^v$. Similarly, let $B : \mc{H}^w \mapsto
\{0,1\}$ be the indicator for $\mc{I}^w$. Thus we have,
\begin{equation}
\E_{x \in \mc{H}^v}[A(x)] \geq \delta/2 \ \ \ \ ,\ \ \ \ 
\E_{y \in \mc{H}^w}[B(y)] \geq \delta/2. \label{eqn-delta}
\end{equation}
Furthermore, since $\mc{I}$ is an independent set, we have,
\begin{equation}
\E_{x,x',y,y'}\left[A(x)A(x')B(y)B(y')\right] = 0,
\end{equation}
where the expectation is according to the distribution induced by
Steps 2, 2a and 2b of the hyperedges construction. 
Expanding out the Fourier expansion of
the above product we obtain,
\begin{equation}
\E_{x,x',y,y'}\left[\sum_{\substack{\alpha,\alpha', \\
\beta,\beta'\subseteq
[m]}}\wh{A}_{\alpha}\wh{A}_{\alpha'}\wh{B}_{\beta}\wh{B}_{\beta'}\chi_{\alpha}(x)\chi_{\alpha'}(x')\chi_\beta(y)\chi_{\beta'}(y')\right]
= 0. \nonumber
\end{equation}
Dropping the subscripts from the expectation and taking it inside
summation,
\begin{equation}
\sum_{\substack{\alpha,\alpha', \\
\beta,\beta'\subseteq
[m]}}\wh{A}_{\alpha}\wh{A}_{\alpha'}\wh{B}_{\beta}\wh{B}_{\beta'}
\E\left[\chi_{\alpha}(x)
\chi_{\alpha'}(x')\chi_\beta(y)\chi_{\beta'}(y')\right] = 0.
\label{eqn-expand0}
\end{equation}
Using the properties of the distribution induced in Steps 2-2b of the
construction, we have the following lemma.
\begin{lemma} \label{lem-zerointer}
Unless $\alpha = \alpha'$, $\beta = \beta'$ and $\pi_{vu}(\alpha)\cap
\pi_{wu}(\beta) = \emptyset$, 
\begin{equation*}
\E\left[\chi_{\alpha}(x)
\chi_{\alpha'}(x')\chi_\beta(y)\chi_{\beta'}(y')\right] = 0.
\end{equation*}
\end{lemma}
\begin{proof}
It can be seen that 
$x_j$ and $x'_{j'}$ are independent for $j' \neq j$, and either one is
independent of $y$ and $y'$. Thus the
expectation vanishes if $\alpha \neq \alpha'$. An identical argument
handles the case when $\beta \neq \beta'$. 

We may assume that $\alpha = \alpha'$ and $\beta = \beta'$. Consider 
the case when $i \in \pi_{vu}(\alpha)\cap\pi_{wu}(\beta)$. From the
construction, in Step 2a, the variables $\{y_j, y'_j | j \in
\pi_{wu}^{-1}(i)\cap\beta\}$ are chosen independently and u.a.r. from
$\{-1,1\}$. Otherwise, in Step 2b, the variables $\{x_j, x'_j | j \in
\pi_{vu}^{-1}(i)\cap\alpha\}$ are chosen independently and u.a.r. from
$\{-1,1\}$. In both cases the expectation vanishes. 
\end{proof}
Observe that  $\pi_{vu}(\alpha)\cap
\pi_{wu}(\beta) = \emptyset$ implies that the variable $\chi_{\alpha}(x)
\chi_{\alpha}(x') = \chi_\alpha(xx')$ is independent of
$\chi_{\beta}(y)\chi_{\beta}(y') = \chi_\beta(yy')$. For convenience
we use the following notation:
\begin{equation} 
 \Gamma^{vu}_\alpha := \E\left[ \chi_\alpha(xx')\right] \ \ \ \ \
\textnormal{and} \ \ \ \ \  
\Gamma^{wu}_\beta := \E\left[ \chi_\beta(yy')\right].
\label{eqn-defgamma}
\end{equation} 
Note that 
$\Gamma^{vu}_\alpha$ and $\Gamma^{wu}_\beta$ depend on the projections
$\pi_{vu}$ and $\pi_{wu}$ respectively. Using Lemma
\ref{lem-zerointer} and Equation
\eqref{eqn-expand0} we obtain,
\begin{equation}
\sum_{\substack{\alpha,\beta \\ \pi_{vu}(\alpha)\cap\pi_{wu}(\beta) =
\emptyset}}\wh{A}_{\alpha}^2\wh{B}_{\beta}^2
\Gamma^{vu}_\alpha\Gamma^{wu}_\beta = 0.
\label{eqn-expand1}
\end{equation}
Using standard arguments along with the fact that $x_i$ is independent
of $x'_j$ for $i \neq j$ we obtain,
\begin{eqnarray}
\E_{x,x'}\left[A(x)A(x')\right] & = &
\sum_{\alpha}\wh{A}_\alpha^2
\E\left[\chi_\alpha(xx')\right] \nonumber \\
& = & \sum_{\alpha}\wh{A}_\alpha^2
\Gamma^{vu}_\alpha. \label{eqn-alphaprod}
\end{eqnarray}
Similarly,
\begin{eqnarray}
\E_{y,y'}\left[B(y)B(y')\right] & = &
\sum_{\beta}\wh{B}_\beta^2
\Gamma^{wu}_\beta. \label{eqn-betaprod}
\end{eqnarray}
To use the above equalities, the goal of the next lemma is to lower bound
$\E[A(x)A(x')]$ and $\E[B(y)B(y')]$. 
\begin{lemma}\label{lem-lowerbd}
For $x,x',y$ and $y'$ as distributed in Steps 2-2b of the construction
of the hyperedges,
$$\E[A(x)A(x')] \geq (\delta/2)^{c_1} \ \ \ \ , \ \ \ \
\E[B(y)B(y')] \geq (\delta/2)^{c_1},$$
where $c_1 = \frac{2\sqrt{2} - 1}{\sqrt{2} - 1}$ is an absolute
constant.
\end{lemma}
\begin{proof} Let us consider $\E[A(x)A(x')]$. The proof for $\E[B(y)B(y')]$ is
analogous. Recall that $A$ is the indicator for
$\mc{I}^v \subseteq \mc{H}^v$. Let $-\mc{I}^v := \{-x \ \mid\
x\in\mc{I}^v\}$.  It is easy to see that,
\begin{equation}
\E\left[A(x)A(x')\right] =  \Pr\left[x \in \mc{I}^v, x'\in
\mc{I}^v\right] 
 =  \Pr\left[x \in \mc{I}^v, -x' \in -\mc{I}^v\right] 
 =   \Pr\left[x \in \mc{I}^v, x'' \in -\mc{I}^v\right], 
\label{eqn-minusx}
\end{equation}
where we use $x''$ to denote $-x'$. 

Consider the product probability space $(\Omega, \mu) = (\prod_{i=1}^k
\Omega_i, \otimes_{i=1}^k \mu_i)$, where for each $i \in [k]$,
$\Omega_i = \{-1,1\}^{\pi_{vu}^{-1}(i)}$ and $\mu_i$ is the uniform
measure. Thus, a
uniformly random $x \in
\mc{H}^v$ (as chosen in Steps 2-2b of the construction) 
can be thought of as belonging to $(\Omega, \mu)$ with
$x|_{\pi_{vu}^{-1}(i)}$ being drawn from $(\Omega_i, \mu_i)$
independently for each $i \in [k]$. In Equation \eqref{eqn-minusx},
both $x$ and $x''$ have uniform marginals distributions. Furthermore,
given $x$, independently for each $i \in [k]$, with probability
$\frac{1}{2}$,
$x''|_{\pi_{vu}^{-1}(i)} =  x|_{\pi_{vu}^{-1}(i)}$ and with
probability $\frac{1}{2}$, $x''|_{\pi_{vu}^{-1}(i)}$ is chosen
uniformly from $(\Omega_i, \mu_i)$. Thus, $x$ and $x''$ are
$\rho$-correlated elements of $(\Omega, \mu)$ with $\rho = \frac{1}{2}$,
according to Definition \ref{def-1}. Since $\mu(\mc{I}^v) = 
\mu(-\mc{I}^v) \geq \delta$, applying Theorem \ref{thm-MOS2} to
Equation \ref{eqn-minusx} we obtain,
\begin{equation}
\E\left[A(x)A(x')\right] =  
\Pr\left[x \in \mc{I}^v, x'' \in -\mc{I}^v\right] \geq
(\delta/2)^{c_1}, \label{eqn-lowerbd}
\end{equation}
which completes the proof of the lemma.
\end{proof}
Using the above lemma along with Equations \eqref{eqn-alphaprod} and
\eqref{eqn-betaprod}   we have,
\begin{equation}
\left(\frac{\delta}{2}\right)^{2c_1} \leq
\left(\sum_{\alpha}\wh{A}_\alpha^2\Gamma^{vu}_\alpha\right) 
\left(\sum_{\beta}\wh{B}_\beta^2\Gamma^{wu}_\beta\right)
 = \sum_{\alpha,
\beta}\wh{A}_\alpha^2\wh{B}_\beta^2\Gamma^{vu}_\alpha\Gamma^{wu}_\beta.
\label{eqn-prodlower}
\end{equation}
Subtracting Equation \eqref{eqn-expand1} from Equation
\eqref{eqn-prodlower}, we obtain,
\begin{equation}
\sum_{\substack{\alpha,\beta \\ \pi_{vu}(\alpha)\cap\pi_{wu}(\beta)
\neq
\emptyset}}\wh{A}_{\alpha}^2\wh{B}_{\beta}^2
\Gamma^{vu}_\alpha\Gamma^{wu}_\beta \geq
\left(\frac{\delta}{2}\right)^{2c_1}. \label{eqn-interlower}
\end{equation}
To continue with the analysis we calculate
$\Gamma^{vu}_\alpha$ and $\Gamma^{wu}_\beta$ in the
following lemma.
\begin{lemma}\label{lem-gamma}
Let, 
$$
\pi_{vu}^{{\sf odd}}(\alpha) := \left\{i \in \pi_{vu}(\alpha) |
\left|\pi_{vu}^{-1}(i)\cap \alpha\right| \textnormal{ is odd}\right\},
\ \ \tn{and,} \ \
\pi_{wu}^{{\sf odd}}(\beta) := \left\{i \in \pi_{wu}(\beta) |
\left|\pi_{wu}^{-1}(i)\cap \beta \right| \textnormal{ is odd}\right\}.
$$
Then,
\begin{equation*}
\Gamma^{vu}_{\alpha} = \frac{(-1)^{\left|\pi_{vu}^{{\sf
odd}}(\alpha)\right|}}{2^{|\pi_{vu}(\alpha)|}}, \ \ \
\textnormal{and,} \ \ \ 
\Gamma^{wu}_{\beta} =  \frac{(-1)^{\left|\pi_{wu}^{{\sf
odd}}(\beta)\right|}}{2^{|\pi_{wu}(\beta)|}}.
\end{equation*}
\end{lemma}
\begin{proof}
From the definition of $\Gamma^{vu}_\alpha$ we can rewrite it as,
\begin{equation}
\Gamma^{vu}_\alpha = \E\left[\prod_{i \in [k]} \left(\prod_{j \in
\pi_{vu}^{-1}(i)\cap \alpha}x_jx'_j\right)\right] 
= \prod_{i \in [k]} \E\left[\prod_{j \in
\pi_{vu}^{-1}(i)\cap \alpha}x_jx'_j\right].\label{eqn-expin}
\end{equation}
For a given $i \in [k]$, with probability $\frac{1}{2}$ all the
variables $x_j, x'_j$ ($j \in \pi_{vu}^{-1}(i)$) are uniformly random
and independent. Otherwise, $x_j$ ($j \in \pi_{vu}^{-1}(i)$)
are chosen independently u.a.r and each $x'_j$ is set to  $-x_j$.
Thus, 
\begin{equation}
\E\left[\prod_{j \in \pi_{vu}^{-1}(i)\cap \alpha}x_jx'_j\right] =
\begin{cases}
\frac{1}{2} & \textnormal{ if } |\pi_{vu}^{-1}(i)\cap \alpha|
\textnormal{ is even.} \\
-\frac{1}{2} & \textnormal{ otherwise.}
\end{cases}
\end{equation}
Substituting the above in Equation \eqref{eqn-expin} proves the lemma
for  $\Gamma^{vu}_\alpha$. The proof for $\Gamma^{wu}_\beta$ is analogous.
\end{proof}
Let $R$ and $T$ ($R > T$) be positive integers to be determined later.
Using the above lemma and Equation \eqref{eqn-interlower} we obtain,
\begin{align}
\sum_{\substack{|\alpha| < R,|\beta| < R \\ \pi_{vu}(\alpha)\cap\pi_{wu}(\beta)
\neq
\emptyset}}\wh{A}_{\alpha}^2\wh{B}_{\beta}^2  + &  
\sum_{\substack{\big[ (|\alpha| \geq R, |\pi_{vu}(\alpha)| < T) \\ \vee 
(|\beta| \geq R, |\pi_{wu}(\beta)| < T)\big], \\ \pi_{vu}(\alpha)\cap\pi_{wu}(\beta)
\neq
\emptyset}}\wh{A}_{\alpha}^2\wh{B}_{\beta}^2 \nonumber \\
 + & 
\sum_{\substack{\big[ (|\alpha| \geq R, |\pi_{vu}(\alpha)| \geq T) \\ \vee 
(|\beta| \geq R, |\pi_{wu}(\beta)| \geq T)\big], \\ \pi_{vu}(\alpha)\cap\pi_{wu}(\beta)
\neq
\emptyset}}\wh{A}_{\alpha}^2\wh{B}_{\beta}^22^{-\left(|\pi_{vu}(\alpha)| +
|\pi_{wu}(\beta)|\right)}
\geq
\left(\frac{\delta}{2}\right)^{2c_1}. \label{eqn-splitlower}
\end{align}
The third term on the LHS of the above
inequality is at most
$2^{-T}\left(\sum_{\alpha}\wh{A}_\alpha^2\right)
\left(\sum_{\beta}\wh{B}_\beta^2\right)\ \leq\ 2^{-T}$, 
using Parseval's identity and the fact that 
$A$ and $B$ are indicator functions.
Similarly, the second term in the LHS of Equation
\eqref{eqn-splitlower} is upper bounded by,
$$\left(\sum_{\substack{|\alpha| \geq R, \\ |\pi_{vu}(\alpha)| <
T}}\wh{A}_\alpha^2\right)\left(\sum_\beta \wh{B}_\beta^2\right) +
\left(\sum_\alpha \wh{A}_\alpha^2\right)\left(\sum_{\substack{|\beta|
\geq R, \\ 
|\pi_{wu}(\beta)| < T}}\wh{B}_\beta^2\right) \leq 
\sum_{\substack{|\alpha| \geq R,\\ \pi_{vu}(\alpha) <
T}}\wh{A}_\alpha^2 + \sum_{\substack{|\beta| \geq R,\\ \pi_{wu}(\beta) <
T}}\wh{B}_\beta^2.$$
Substituting the above in Equation \eqref{eqn-splitlower} we obtain
that for any two good vertices $v, w \in V$ which share a neighbor $u
\in U$, 
\begin{equation}
\sum_{\substack{|\alpha| < R,|\beta| < R \\ \pi_{vu}(\alpha)\cap\pi_{wu}(\beta)
\neq
\emptyset}}\wh{A}_{\alpha}^2\wh{B}_{\beta}^2  
+ \sum_{\substack{|\alpha| \geq R,\\ |\pi_{vu}(\alpha)| <
T}}\wh{A}_\alpha^2 + \sum_{\substack{|\beta| \geq R,\\
|\pi_{wu}(\beta)| <
T}}\wh{B}_\beta^2\ \geq\ \left(\frac{\delta}{2}\right)^{2c_1} - 2^{-T}.
\label{eqn-fingood}
\end{equation}
Consider the following process of selecting $u, v$ and $w$. Choose $u$
u.a.r from $U$ and $v$ and $w$ be two neighbors of $u$ chosen
independently and u.a.r from its neighborhood. Let $p_u$ be the
fraction of the neighbors of $u$ that are ``good''. Since, $\delta/2$
fraction of the vertices in $V$ are good and the graph $G(U, V, E)$ is
bi-regular, $\E_{u\in U}[p_u] \geq (\delta/2)$. Thus, the probability
that both $v$ and $w$ are ``good'' is $\E_u[p_u^2] \geq
\left(\E_u[p_u]\right)^2 \geq (\delta/2)^2$. Taking an expectation
over the choice of $u, v$ and $w$, and noting that the LHS of Equation
\eqref{eqn-fingood} is always positive, we obtain,
\begin{equation}
\E_{u,v,w}\left[ 
\sum_{\substack{|\alpha| < R,|\beta| < R \\ \pi_{vu}(\alpha)\cap\pi_{wu}(\beta)
\neq
\emptyset}}\wh{A}_{\alpha}^2\wh{B}_{\beta}^2 
+ \sum_{\substack{|\alpha| \geq R,\\ \pi_{vu}(\alpha) <
T}}\wh{A}_\alpha^2 + \sum_{\substack{|\beta| \geq R,\\ \pi_{wu}(\beta) <
T}}\wh{B}_\beta^2\right]  
\geq
\left(\frac{\delta}{2}\right)^2\left[\left(\frac{\delta}{2}\right)^{2c_1}
- 2^{-T}\right]. \label{eqn-expectbd}
\end{equation}
In order to bound the second and third terms in the above expectation
we use property (b) in Theorem \ref{thm-LC}. For a fixed
vertex $v \in V$ and subset $\alpha \subseteq [m]$ such that $|\alpha|
\geq R$,
\begin{equation}
\Pr_u\left[\left|\pi_{vu}(\alpha)\right| < R^{c_0}\right] \ \leq\
\frac{1}{R^{c_0}},
\end{equation}
where the probability is over a random neighbor $u$ of $v$. Thus,
\begin{equation}
\E_u\left[ \sum_{\substack{|\alpha| \geq R,\\ |\pi_{vu}(\alpha)| <
R^{c_0}}}\wh{A}_\alpha^2\right] =  \sum_{|\alpha| \geq R}
\left[\wh{A}_\alpha^2 \cdot \Pr_u\left[\left|\pi_{vu}(\alpha)\right| <
R^{c_0}\right]\right] \leq  \sum_{|\alpha| \geq R} \wh{A}_\alpha^2\cdot
\frac{1}{R^{c_0}}\ \leq\ \frac{1}{R^{c_0}}. \label{eqn-pilargebd}
\end{equation}
Setting $T = R^{c_0}$ and substituting the above into Equation
\eqref{eqn-expectbd} we obtain,
\begin{equation}
\E_{u,v,w}\left[ 
\sum_{\substack{|\alpha| < R,|\beta| < R \\ \pi_{vu}(\alpha)\cap\pi_{wu}(\beta)
\neq
\emptyset}}\wh{A}_{\alpha}^2\wh{B}_{\beta}^2 \right] 
\geq
\left(\frac{\delta}{2}\right)^2\left[\left(\frac{\delta}{2}\right)^{2c_1}
- 2^{-R^{c_0}}\right] - \frac{2}{R^{c_0}}. \label{eqn-finexpbd}
\end{equation}
Let $c' = 2 + 2c_1$. Setting $R = 8/(\delta/2)^{c'/c_0}$ and using $2^{-R^{c_0}} \leq R^{-c_0}$ in the above
inequality yields,
\begin{equation}
\E_{u,v,w}\left[ 
\sum_{\substack{|\alpha|, |\beta| < 8/(\delta/2)^{c'/c_0} 
\\ \pi_{vu}(\alpha)\cap\pi_{wu}(\beta)
\neq
\emptyset}}\wh{A}_{\alpha}^2\wh{B}_{\beta}^2 \right] \geq
\frac{1}{4}\left(\frac{\delta}{2}\right)^{c'}. \label{eqn-finbd}
\end{equation}

\medskip
\noindent
{\bf Labeling.} The above analysis yields the following randomized
labeling $\sigma$ of the vertices of $\mc{L}$. For a vertex $v \in V$, choose
a subset $\alpha$ probability $\wh{A}_\alpha^2$ and assign as label
$\sigma(v)$ a
random $i \in \alpha$. For any vertex $u \in U$, randomly choose a
neighbor $w$ and assign $\pi_{wu}(\sigma(w))$ as the label to $u$.
Equation \eqref{eqn-finbd} implies that the expected fraction of
constraints satisfied is at least,
\begin{equation}
\frac{1}{256}\left(\frac{\delta}{2}\right)^{c' + 2c'/c_0}.
\end{equation}

\subsubsection{Choice of parameters}
In Theorem \ref{thm-LC} we can choose $r = (\log \log N)/4$. This
ensures that the instance $\mc{G}$ is of size $n = N^{O(r)}2^{2^{3r}} \leq
N^{O(\log\log N)}$. The soundness of $\mc{L}$ is $2^{-\Omega(\log\log
N)} = 2^{-\Omega(\log\log n)}$. Combining this with the above analysis in
the NO Case, choosing $\delta = \frac{1}{(\log n)^c}$ for some
positive constant $c$ (depending on $c_0, c_1, \gamma_0$) we obtain a
contradiction to our  assumption on the size of the independent set. 

Thus, in the NO Case, $\mc{G}$ does not contain independent set of
$\frac{1}{(\log n)^c}$ relative size. This completes the proof of Theorem
\ref{thm-main1}.

\section{Satisfiable {\sc Max-E$3$-SAT}}\label{sec-max3}
As before, the input is an instance $\mc{L}$ of \LC from Theorem
\ref{thm-LC} consisting of a
bipartite graph $G(U, V, E)$, label sets $[m]$ and $[k]$ and
projections $\{\pi_{vu} : [m] \mapsto [k] \mid \{u,v\}\in E, u\in U,
v\in V\}$.

The PCP proof is the same as in \cite{Hastad}.
For each vertex $u \in U$ there is a Long Code $\mc{H}^u =\{-1,1\}^k$. 
Similarly, for each $v \in V$, there is $\mc{H}^v = \{-1,1\}^m$. The
assignments to these Long Codes are $A^u : \mc{H}^u \mapsto \{-1,1\}$
and $B^v : \mc{H}^v \mapsto \{-1,1\}$. We can
assume that these
assignments are \emph{folded} over $-1$, i.e. $A^u(x) = -A^u(-x)$ and $B^v(y)
= -B^v(-y)$.

The instance of {\sc Max-E$3$-SAT} is given by the following PCP verifier
whose acceptance predicate corresponds to a $3$-literal clause. Let
$\eps > 0$ be a parameter which we shall set later. 

\medskip
\noindent
{\bf PCP Verifier}
\begin{itemize}
\item[1.] Choose a vertex $u \in U$ u.a.r and one
of its neighbors $v\in V$ u.a.r. 

\item[2.] Choose $x \in \mc{H}^u$ u.a.r.

\item[3.] Let $y, y' \in \mc{H}^v$ be
chosen
as follows. For each $i \in [k]$, if
$x_i = 1$ do
Step 3 otherwise do Step 4.

\item[4.] $x_i = 1$: Independently 
for each $j \in \pi_{vu}^{-1}(i)$ choose $y_j$
u.a.r from $\{-1,1\}$ and set $y'_j = -y_j$. 

\item[5.] $x_i = -1$: Do Step 5a with probability $1-\eps$, or Step 5b
with probability $\eps$.
\begin{itemize}
\item[5a.] Independently for each $j \in \pi_{vu}^{-1}(i)$ choose 
$y_j$ u.a.r. from $\{-1,1\}$ and set $y'_j$ to be $y_j$.

\item[5b.] Independently for each $j \in  \pi_{vu}^{-1}(i)$, choose $y_j$ 
and $y'_j$ independently and u.a.r. from $\{-1,1\}$. 
\end{itemize}

\item[6.] Accept if $\left(A^u(x), B^v(y), B^v(y')\right) \neq
(1,1,1)$.
\end{itemize}
The above PCP predicate (after folding) is equivalent -- in terms of
its completeness and soundness -- to a gap
instance of {\sc Max-E$3$-SAT}.

\subsection{Analysis: YES Case}
In the YES case there is a labeling $\sigma$ to the vertices of
$\mc{L}$ that satisfies all the constraints. Consider the
assignment $A^u(x) = x_{\sigma(u)}$ and similarly $B^{v}(y) =
y_{\sigma(v)}$ for all $u \in U$ and $v \in V$. Clearly, these
assignments are folded over $-1$. Furthermore, in the choice of $x, y,
y'$ in the PCP test, it is easy to see that $(x_{\sigma(u)},
y_{\sigma(v)}, y'_{\sigma(v)}) \neq (1,1,1)$ since
$\pi_{vu}(\sigma(v)) = \sigma(u)$. Thus, the PCP test is always
satisfied and there is an assignment that satisfies all the
clauses of the corresponding {\sc Max-E$3$-SAT} instance.

For notational simplicity in the rest of the analysis we shall drop
the superscripts to denote $A^u$ by $A$ and $B^v$ by $B$.

\subsection{Analysis: NO Case}
Suppose for a contradiction that,
\begin{equation}\displaystyle
\E_{\substack{u,v\\x,y,y'}}\left[1 -
\frac{(1+A(x))(1+B(y))(1+B(y'))}{8}\right] \geq \frac{7}{8} + \delta,
\label{eqn-3satno}
\end{equation}
where the expectation is over the choices of the verifier and thus the
LHS denotes the probability that the verifier accepts. We shall show
that (for an appropriate setting of $\eps$) there is a labeling to 
the vertices of $\mc{L}$ that 
satisfies $\delta^{O(1)}$ fraction of edges. 
This is in contrast to the PCP test in \cite{Hastad} which
yields a bound of $\delta^{O(\delta^{-1})}$. 

In the following analysis we fix the choice of $u$ and $v$ for the time
being.  

Since the Long Codes are folded, we have $\E[A(x)] = \E[B(y)] = 
\E[B(y')] = 0$, as the distributions of $x \in \mc{H}^u$ and $y,
y' \in \mc{H}^v$ are respectively uniform. Further, $x$ is independent of
$y$ and independent of $y'$, and thus $\E[A(x)B(y)] = \E[A(x)B(y')] =
\E[A(x)]\E[B(y)] = 0$. The rest of the terms are
analyzed as follows.

\begin{lemma}\label{lem-bb1} 
$\left|\E[B(y)B(y')]\right| \leq \eps/2.$\end{lemma}
\begin{proof} Using the Fourier expansion of $B$, and since $B$ is
folded,
\begin{align} 
\E[B(y)B(y')]  = & \sum_{|\beta|
\textnormal{odd}}\wh{B}_\beta^2\E\left[\chi_\beta(yy')\right]
\nonumber \\
= &  \sum_{|\beta|
\textnormal{ odd}}\wh{B}_\beta^2\prod_{i \in
[k]}\E\left[\chi_{\beta\cap \pi_{vu}^{-1}(i)}(yy')\right] \label{eqn-BB1}.
\end{align}
For an odd sized $\beta$, there is a $i \in [k]$ such that $\left|\beta\cap
\pi_{vu}^{-1}(i)\right|$ is odd. It is easy to check that for such a
$i$, 
$\E\left[\chi_{\beta\cap \pi_{vu}^{-1}(j)}(yy')\right] = -\eps/2$. Also
note
that for any $i$, $\left|\E\left[\chi_{\beta\cap
\pi_{vu}^{-1}(i)}(yy')\right]\right| \leq 1$. Thus, Equation
\eqref{eqn-BB1} yields,
$$\left|\E[B(y)B(y')]\right| \leq (\eps/2)\sum_{|\beta| \textnormal{ odd}}
\wh{B}_\beta^2 = \eps/2.$$\end{proof}
 
\begin{lemma}\label{lem-RT} For any positive integers $R, T$ such that 
$R \geq T$,
$$\left|\E[A(x)B(y)B(y')]\right| \leq \left(\sum_{\substack{|\alpha|, |\beta| \tn{ odd} \\
\alpha \subseteq \pi_{vu}(\beta) \\ |\beta| < R}}
\wh{A}_\alpha^2\wh{B}_\beta^2\right)^{\frac{1}{2}}
+  \sum_{\substack{|\beta| \geq R \\
\left|\pi_{vu}(\beta)\right| < T}} 
\wh{B}_\beta^2
+ \left(1 - \frac{\eps}{2}\right)^{\frac{T}{2}}.$$
\end{lemma}
\begin{proof}
Using the fact $y_j$ ($y'_j$) is independent of $x_i$ for any $i$, the
term $\E\left[\chi_\alpha(x)\chi_\beta(y)\chi_{\beta'}(y')\right]$ is zero 
unless $\beta = \beta'$ and $\alpha \subseteq \pi_{vu}(\beta)$. Thus,
\begin{align}
\E[A(x)B(y)B(y')] =  & \sum_{\substack{|\alpha|, |\beta| \tn{ odd} \\
\alpha \subseteq \pi_{vu}(\beta)}}
\wh{A}_\alpha\wh{B}_\beta^2\E\left[\chi_\alpha(x)\chi_\beta(yy')\right]
\nonumber \\
 = &  \sum_{\substack{|\alpha|, |\beta| \tn{ odd} \\
\alpha \subseteq \pi_{vu}(\beta)}} 
\wh{A}_\alpha\wh{B}_\beta^2 \left(\prod_{i\in \alpha}
\E\left[x_i\chi_{\beta\cap
\pi_{vu}^{-1}(i)}(yy')\right]\right) \nonumber \\ 
 & \ \ \ \ \ \ \ \ \ \ \left(\prod_{i \in
\pi_{vu}(\beta)\setminus \alpha} 
\E\left[\chi_{\beta\cap
\pi_{vu}^{-1}(i)}(yy')\right]\right). \label{eqn-brack}
\end{align}
To simplify the above equation we require the following lemma.
\begin{lemma} \label{lem-measure}
Fix $\beta \neq \emptyset$ and let $r = \left|\pi_{vu}(\beta)\right|$.
For any $\alpha \subseteq
\pi_{vu}(\beta)$ let,
$$p_\beta(\alpha) := \left|\prod_{i\in \alpha}
\E\left[x_i\chi_{\beta\cap
\pi_{vu}^{-1}(i)}(yy')\right]\right|\cdot
\left|\prod_{i \in
\pi_{vu}(\beta)\setminus \alpha} 
\E\left[\chi_{\beta\cap
\pi_{vu}^{-1}(i)}(yy')\right]\right|.$$
Then, 
$$p_\beta(\alpha) =
\left(\frac{\eps}{2}\right)^{r'}\left(1 - \frac{\eps}{2}\right)^{r -
r'},$$
where $r' = \left|\alpha \Delta \pi_{vu}^{{\sf odd}}(\beta)\right|$
and $\pi_{vu}^{{\sf odd}}$ is as defined in Lemma \ref{lem-gamma}.
Thus, $p_\beta(\alpha)$ is a probability measure over $\alpha
\subseteq \pi_{vu}(\beta)$. 
\end{lemma}
\begin{proof}
It is easy to verify that for any $i \in [k]$ and $J \subseteq
\pi_{vu}^{-1}(i)$, $J\neq \emptyset$, $|J|$ even, 
\begin{equation}
\left|\E\left[x_i\chi_{J}(yy')\right]\right| = \frac{\eps}{2}, \ 
\ \ \ \ 
\left|\E\left[\chi_{J}(yy')\right]\right| = 1-\frac{\eps}{2}.
\end{equation} 
Similarly, for $i \in [k]$ and $J \subseteq
\pi_{vu}^{-1}(i)$, $|J|$ odd,
\begin{equation}
\left|\E\left[x_i\chi_{J}(yy')\right]\right| = 1 - \frac{\eps}{2}, \
\ \ \ \  
\left|\E\left[\chi_{J}(yy')\right]\right| = \frac{\eps}{2}.
\end{equation}
The above equations imply,
\begin{equation}
\left|\prod_{i\in \alpha}
\E\left[x_i\chi_{\beta\cap
\pi_{vu}^{-1}(i)}(yy')\right]\right| =
\left(1-\frac{\eps}{2}\right)^{\left|\alpha\cap \pi^{{\sf
odd}}_{vu}(\beta)\right|}\left(\frac{\eps}{2}\right)^{
\left|\alpha\cap \ol{\pi^{{\sf
odd}}_{vu}(\beta)}\right|},
\end{equation}
and,
\begin{equation}
\left|\prod_{i \in
\pi_{vu}(\beta)\setminus \alpha} 
\E\left[\chi_{\beta\cap
\pi_{vu}^{-1}(i)}(yy')\right]\right| = 
\left(1-\frac{\eps}{2}\right)^{\left|\ol{\alpha}\cap\ol{\pi^{{\sf
odd}}_{vu}(\beta)}\right|}\left(\frac{\eps}{2}\right)^{
\left|\ol{\alpha}\cap\pi^{{\sf
odd}}_{vu}(\beta)\right|},
\end{equation}
where, in the above two equations $\ol{\ \cdot\ }$ denotes the
$\pi_{vu}(\beta)\setminus . $ operation. Combining them we obtain the
lemma.
\end{proof}
Using the above in Equation \eqref{eqn-brack} we obtain,
\begin{equation}
\left|\E[A(x)B(y)B(y')]\right| 
\leq
\sum_{\substack{|\alpha|, |\beta| \tn{ odd} \\
\alpha \subseteq \pi_{vu}(\beta)}}
\left|\wh{A}_\alpha\wh{B}_\beta^2\right|p_\beta(\alpha).
\end{equation}
A categorization of the terms of the above inequality based on the
parameters $R$ and $T$ yields,
\begin{align}
\left|\E[A(x)B(y)B(y')]\right| 
\leq & 
\sum_{\substack{|\alpha|, |\beta| \tn{ odd} \\
\alpha \subseteq \pi_{vu}(\beta) \\ |\beta| < R}}
\left|\wh{A}_\alpha\wh{B}_\beta^2\right|
p_\beta(\alpha)
+
\sum_{\substack{|\alpha|, |\beta| \tn{ odd} \\
\alpha \subseteq \pi_{vu}(\beta) \\ |\beta| \geq R \\
\left|\pi_{vu}(\beta)\right| < T}} 
\left|\wh{A}_\alpha\wh{B}_\beta^2\right|p_\beta(\alpha) \nonumber
\\
& \ \ \ + 
\sum_{\substack{|\alpha|, |\beta| \tn{ odd} \\
\alpha \subseteq \pi_{vu}(\beta) \\ |\beta| \geq R \\
\left|\pi_{vu}(\beta)\right| \geq T}} 
\left|\wh{A}_\alpha\wh{B}_\beta^2\right|p_\beta(\alpha)\label{eqn-absplit}
\end{align}
For each of the three terms in the RHS above, we apply Cauchy-Schwarz in the
following manner. Each $\wh{B}_\beta^2$ is multiplied by,
\begin{align}
\sum_{\substack{\alpha \subseteq \pi_{vu}(\beta)\\ |\alpha| \tn{
odd}}}\left|\wh{A}_\alpha\right|
p_\beta(\alpha) \leq &  
\sum_{\substack{\alpha \subseteq \pi_{vu}(\beta)\\ |\alpha| \tn{
odd}}}\left|\wh{A}_\alpha\right|
\sqrt{p_\beta(\alpha)}\sqrt{p_\beta(\alpha)} \nonumber \\
\leq & \left( \sum_{\substack{\alpha \subseteq \pi_{vu}(\beta)\\ 
|\alpha| \tn{ odd}}}
\wh{A}_\alpha^2\right)^{\frac{1}{2}}\left(
 \sum_{\alpha \subseteq \pi_{vu}(\beta)}
p_\beta(\alpha) \right)^{\frac{1}{2}}\left(\max_{\alpha \subseteq
\pi_{vu}(\beta)}\sqrt{p_\beta(\alpha)}\right)  \label{eqn-cauchytwice}
\end{align}
By Parseval's and since
$p_\beta(\alpha)$ is a probability measure over $\alpha$ we obtain that the
RHS of Equation \eqref{eqn-cauchytwice} is bounded by $1$. Thus, the
second term in the RHS of Equation \eqref{eqn-absplit} is bounded by,
\begin{equation}
\sum_{\substack{|\beta| \geq R \\
\left|\pi_{vu}(\beta)\right| < T}} 
\wh{B}_\beta^2. \label{eqn-secterm}
\end{equation}
Further, observe 
that for
$\beta$ such that
$\pi_{vu}(\beta) \geq T$, 
$$p_\beta(\alpha) \leq \left(1 - \frac{\eps}{2}\right)^{T},$$
for any $\alpha \subseteq \pi_{vu}(\beta)$. Thus, the third term in
the RHS of 
Equation \eqref{eqn-absplit} is bounded by,
\begin{equation}
\left(1 - \frac{\eps}{2}\right)^{\frac{T}{2}}\sum_{\beta}\wh{B}_\beta^2 \leq 
\left(1 - \frac{\eps}{2}\right)^{\frac{T}{2}}. \label{eqn-thirdterm}
\end{equation}
For the first term in the RHS of Equation  \eqref{eqn-absplit}, we use
Equation \eqref{eqn-cauchytwice} to obtain the following upper bound.
\begin{align}
\sum_{\substack{|\alpha|, |\beta| \tn{ odd} \\
\alpha \subseteq \pi_{vu}(\beta) \\ |\beta| < R}}
\left|\wh{A}_\alpha\wh{B}_\beta^2\right|
p_\beta(\alpha) \leq &
\sum_{|\beta| \tn{ odd}, |\beta| < R}
 \left( \sum_{\substack{\alpha \subseteq \pi_{vu}(\beta)\\ 
|\alpha| \tn{ odd}}}
\wh{A}_\alpha^2\right)^{\frac{1}{2}}\wh{B}_\beta^2 \nonumber \\
\leq & \left(\sum_{|\beta| \tn{ odd}, |\beta| < R}
 \left( \sum_{\substack{\alpha \subseteq \pi_{vu}(\beta)\\ 
|\alpha| \tn{ odd}}}
\wh{A}_\alpha^2\right)\wh{B}_\beta^2\right)^{\frac{1}{2}}\left(
\sum_{|\beta| \tn{ odd}, |\beta| < R}
\wh{B}_\beta^2\right)^{\frac{1}{2}} \nonumber \\
\leq & \left(\sum_{\substack{|\alpha|, |\beta| \tn{ odd} \\
\alpha \subseteq \pi_{vu}(\beta) \\ |\beta| < R}}
\wh{A}_\alpha^2\wh{B}_\beta^2\right)^{\frac{1}{2}},
\label{eqn-firstterm}
\end{align}
where the second last inequality above is obtained by an application
of Cauchy-Schwarz and the last by Parseval's.
Substituting Equations \eqref{eqn-secterm}, \eqref{eqn-thirdterm} and
\eqref{eqn-firstterm}
into \eqref{eqn-absplit}  
completes the proof of the Lemma \ref{lem-RT}.
\end{proof}

Setting $T = R^{c_0}$, taking the expectation over a random neighbor
$u$ of a fixed $v$ in Lemma \ref{lem-RT} and using the analysis of Equation
\eqref{eqn-pilargebd}, we obtain,
\begin{equation}
\E_{u}\left[\left|\E_{x,y,y'}[A(x)B(y)B(y')]\right|\right] \leq 
\E_u\left[\left(\sum_{\substack{|\alpha|, |\beta| \tn{ odd} \\
\alpha \subseteq \pi_{vu}(\beta) \\ |\beta| < R}}
\wh{A}_\alpha^2\wh{B}_\beta^2\right)^{\frac{1}{2}}\right]
+  \frac{1}{R^{c_0}}
+ \left(1 - \frac{\eps}{2}\right)^{\frac{R^{c_0}}{2}}.
\end{equation}
Using the above, Lemma \ref{lem-bb1} and Equation
\eqref{eqn-3satno} we obtain,
\begin{equation}
\E_{u,v}\left[\left(\sum_{\substack{|\alpha|, |\beta| \tn{ odd} \\
\alpha \subseteq \pi_{vu}(\beta) \\ |\beta| < R}}
\wh{A}_\alpha^2\wh{B}_\beta^2\right)^{\frac{1}{2}}\right]
\geq 8\delta - \frac{\eps}{2} -  \frac{1}{R^{c_0}}
- \left(1 - \frac{\eps}{2}\right)^{\frac{R^{c_0}}{2}}.
\end{equation}
Applying Cauchy-Schwarz and setting $R =
\left(\frac{4}{\eps}\log\left(\frac{1}{\eps}\right)\right)^{\frac{1}{c_0}}$,
simplifies the above to,
\begin{equation}
\left(\E_{u,v}\left[\sum_{\substack{|\alpha|, |\beta| \tn{ odd} \\
\alpha \subseteq \pi_{vu}(\beta) \\ |\beta| < R}}
\wh{A}_\alpha^2\wh{B}_\beta^2\right]\right)^{\frac{1}{2}}
\geq 8\delta - 2\eps.\end{equation}
Finally, we set $\eps = \delta$ to make the RHS of the above at least
$6\delta$. This yields a labeling for the \LC instance $\mc{L}$: for
every vertex $u \in U$ uniformly choose a subset $\alpha$ of labels with
probability $\wh{A^u}_\alpha^2$ and assign it a random label from $\alpha$.
Similarly, for every vertex $v \in V$ uniformly select a set of labels
$\beta$ with probability $\wh{B^v}_\beta^2$ and assign it a random label
from $\beta$. The above analysis shows that the expected fraction of
edges satisfied is,
$$ \frac{36\delta^2}{R} = \Omega\left(\delta^{c'}\right),$$
for some positive constant $c'$ depending on $c_0$.

\subsubsection{Choice of parameters}
In Theorem \ref{thm-LC} we can choose $r = (\log \log N)/4$. This
ensures that the reduction to {\sc Max-E$3$-SAT} 
is of size $n = N^{O(r)}2^{2^{3r}} \leq
N^{O(\log\log N)}$. The soundness of $\mc{L}$ is $2^{-\Omega(\log\log
N)} = 2^{-\Omega(\log\log n)}$. Combining this with the above analysis in
the NO Case, choosing $\delta = \frac{1}{(\log n)^c}$ for some
positive constant $c$ (depending on $c_0$ and $\gamma_0$) we obtain a
contradiction to our assumption on the probability of acceptance of the 
verifier. 

Thus, in the NO Case, the verifier accepts with probability at most
$\frac{7}{8} + \frac{1}{(\log n)^c}$. This completes the proof of Theorem
\ref{thm-main2}.

\section{Satisfiable {\sc Max-E$4$-Set-Splitting}}\label{sec-4ss}
As in the previous sections, the input is an instance $\mc{L}$ of 
\LC from Theorem \ref{thm-LC} consisting of a
bipartite graph $G(U, V, E)$, label sets $[m]$ and $[k]$ and
projections $\{\pi_{vu} : [m] \mapsto [k] \mid \{u,v\}\in E, u\in U,
v\in V\}$.

The PCP proof is similar to the previous sections.
For each vertex $v \in V$, there is a Long Code $\mc{H}^v = \{-1,1\}^m$. The
assignments to these Long Codes are $A^v : \mc{H}^v \mapsto \{-1,1\}$. 
In the case of {\sc Max-E$4$-Set-Splitting} we \emph{do not} have folding of the
Long Codes.

The instance of {\sc Max-E$4$-Set-Splitting} is given by the following 
PCP verifier whose $4$-query tests correspond to the $4$-sets of the
instance. The rejection probability of the predicate estimates 
the fraction $4$-query tests completely 
contained in the subset corresponding to the $1$s of the 
proof locations. Let
$\eps > 0$ be a parameter which shall be set later. 

\medskip
\noindent
{\bf PCP Verifier.}

\begin{itemize}
\item[1.] Choose a vertex $u \in U$ u.a.r and two of its
 neighbors $v, w\in V$ independently and u.a.r. 

\item[2.] Choose $x \in \mc{H}^v$ and $y \in \mc{H}^w$ independently
and u.a.r.

\item[3.] For each $i \in [k]$, either do Step 3a or Step 3b with
probability $\frac{1}{2}$ each.
\begin{itemize}
\item[3a.] For each $j \in \pi_{vu}^{-1}(i)$ set $x'_j = -x_j$.
Further, with probability $1-\eps$ do Step 3a.1, or Step 3a.2 with
probability $\eps$.
\begin{itemize}
\item[3a.1] For each $j \in \pi_{wu}^{-1}(i)$ set $y'_j = y_j$.
\item[3a.2] For each $j \in \pi_{wu}^{-1}(i)$ independently, set 
$y'_j$ u.a.r from
$\{-1,1\}$.
\end{itemize}
\end{itemize}

\begin{itemize}
\item[3b.] For each $j \in \pi_{wu}^{-1}(i)$ set $y'_j = -y_j$.
Further, with probability $1-\eps$ do Step 3b.1, or Step 3b.2 with
probability $\eps$.
\begin{itemize}
\item[3b.1] For each $j \in \pi_{vu}^{-1}(i)$ set $x'_j = x_j$.
\item[3b.2] For each $j \in \pi_{vu}^{-1}(i)$ independently, 
set $x'_j$ u.a.r from
$\{-1,1\}$.
\end{itemize}
\end{itemize}
\item[4.] Reject if $(A^v(x), A^v(x'), A^w(y), A^w(y')) = (1,1,1,1)$.
\end{itemize}

The above PCP verifier is equivalent -- in terms of
its completeness and soundness -- to a gap
instance of {\sc Max-E$4$-Set-Splitting}.
\subsection{Analysis: YES Case}
In the YES case there is a labeling $\sigma$ to the vertices $\mc{L}$ 
that satisfies all its edges. Consider the
assignment $A^v(x) = x_{\sigma(v)}$ for all $v \in V$. For the choice
of $u, v$, and $w$ in the above PCP we have $\pi_{vu}(\sigma(v)) =
\pi_{wu}(\sigma(w))$. Thus, from the choice of $x, x', y,
y'$ in the PCP test, it is easy to see that $(x_{\sigma(v)},
x'_{\sigma(v)},
y_{\sigma(w)}, y'_{\sigma(w)}) \not\in \{ (1,1,1,1), (-1,-1,-1,-1)\}$. 
Thus, the PCP test is always
satisfied and the there is an assignment that splits all the $4$-sets
of the {\sc Max-$E$4-Set-Splitting} instance.

For notational simplicity in the rest of the analysis we shall drop
the superscripts to denote $A^v$ by $A$ and $A^w$ by $B$.

\subsection{Analysis: NO Case}
The probability that the PCP verifier rejects is given by,
\begin{equation}
\frac{1}{16}\cdot\E\left[(1+A(x))(1+A(x'))(1+B(y))(1+B(y'))\right],\nonumber
\end{equation}
where the expectation is above the random choice of $u, v, w, x, x',
y$ and $y'$ by the PCP verifier. Expanding the above we obtain that
the probability of rejection of the verifier is ,
\begin{align}
\frac{1}{16}\cdot\E\bigg[1 + A(x) + A(x') + B(x) +  & B(y') + 
A(x)A(x') +  A(x)B(y)  + A(x)B(y') \nonumber \\
+ A(x')B(y) 
 + A(x')&B(y')  +  B(y)B(y') + A(x)A(x')B(y) \nonumber \\
 + A(x')B(y)B(y') & +  A(x)B(y)B(y') + A(x)A(x')B(y') \nonumber
\\ 
& \ \ \ \ \ \ \ +  A(x)A(x')B(y)B(y')\bigg].
\label{eqn-4ss-acc}
\end{align}
Let the number of $1$s in the proof be exactly $\rho$ fraction, i.e.,
\begin{equation}
\E_{v, x}\left[A^v(x)\right] = 2\rho - 1. \label{eqn-rho-frac}
\end{equation}
where the expectation is over a random vertex $v \in V$ and a
uniformly chosen $x \in \mc{H}^v$. 
Assume that the probability that the verifier
rejects is at most $\rho^4 - \delta$ for some $\delta > 0$. 
We shall show that this implies a
labeling to the vertices of $\mc{L}$ that satisfies
$\delta^{O(1)}$
fraction of edges (using an appropriate choice of $\eps > 0$
depending only on $\delta$). For the analysis we shall consider
the terms in the expectation in Equation \eqref{eqn-4ss-acc} one by one. 
Before proceeding, we fix the choice
of $u$ in the expectation for the time being. Let $p_u :=
\E_{v\sim u, x\in \mc{H}^v}[A^v(x)]$, where the expectation is over a random
neighbor $v$ of $u$.

Since $v$ and $w$ are u.a.r neighbors of $u$, and by the
uniformity of $x, x', y$ and $y'$,
\begin{equation}
\E[A(x)] = \E[A(x')] = \E[B(y)] = \E[B(y')] = p_u. \label{eqn-4ss-one}
\end{equation}

Observe that $x$ is independent of $y$ and of $y'$. For a fixed choice
of $u$, $v$ and $w$ are two independently chosen random neighbors of
$u$. This implies that,
\begin{equation*}
\E\left[A(x)B(y)\right] = \left(\E_{v, x\in
\mc{H}^v}\left[A^v(x)\right]\right)^2 = p_u^2. 
\end{equation*}
This also holds for the other \emph{cross} terms and thus,
\begin{equation}
\E\left[A(x)B(y)\right] = \E\left[A(x)B(y')\right] = 
\E\left[A(x')B(y)\right] = \E\left[A(x')B(y')\right] = p_u^2. 
\label{eqn-4ss-cross}
\end{equation}

Before analyzing the rest of the terms we require the following
lemmas. Fix the choice of $v$ and $w$ for the next two lemmas.
\begin{lemma}\label{lem-4ss-J}
Let $i \in [k]$ and $J \subseteq \pi_{vu}^{-1}(i)$, be non-empty.
Then,
\begin{equation}
\E\left[\chi_J(xx')\right] = \begin{cases}
					1 - \frac{\eps}{2}
&\tn{ if }|J|\tn{ even.} \\
		-\frac{\eps}{2} &\tn{ if }|J|\tn{ odd.}
			\end{cases}\label{eqn-4ss-J}
\end{equation}
A similar property holds for $\pi_{wu}$ with $y$ and $y'$.
\end{lemma}
\begin{proof} Note that in the choice of the verifier
$x'|_{\pi_{vu}^{-1}(i)}$ is chosen to be $-x|_{\pi_{vu}^{-1}(i)}$ with
probability $\frac{1}{2}$, $x|_{\pi_{vu}^{-1}(i)}$ with probability
$\frac{1-\eps}{2}$, and u.a.r with probability $\frac{\eps}{2}$. The
lemma follows, and holds analogously for $\pi_{wu}$ with $y$ and $y'$.
\end{proof}
The above immediately implies the following lemma,
\begin{lemma}\label{lem-4ss-oddeven}
Let $\alpha \subseteq [m]$, and $r = \left|\pi_{vu}(\alpha)\right|$
and $r' = \left|\pi_{vu}^{{\sf odd}}(\alpha)\right|$ (as per the
definition in Lemma \ref{lem-gamma}). Then,
\begin{equation}
\E\left[\chi_\alpha(xx')\right] =
\left(1-\frac{\eps}{2}\right)^{r-r'}\left(-\frac{\eps}{2}\right)^{r'}.
\end{equation}
Similarly, for $\beta \subseteq [m]$, $r =
\left|\pi_{wu}(\beta)\right|$ and $r' = \left|\pi_{wu}^{{\sf
odd}}(\beta)\right|$,
\begin{equation}
\E\left[\chi_\beta(yy')\right] =
\left(1-\frac{\eps}{2}\right)^{r-r'}\left(-\frac{\eps}{2}\right)^{r'}.
\end{equation}
\end{lemma}
We are now ready to bound the terms $\E\left[A(x)A(x')\right]$ and
$\E\left[B(y)B(y')\right]$, where the choice of $u$ is fixed. 
\begin{lemma}\label{lem-4ss-xx}
$\E_{v,x,x'}\left[A(x)A(x')\right] = \E_{w,y,y'}\left[B(y)B(y')\right] \geq
p_u^2 -\eps/2.$
\end{lemma}
\begin{proof} Using the Fourier expansion along with standard
arguments we have,
\begin{equation}
\E_{v,x,x'}\left[A(x)A(x')\right] = \E_v\left[ \wh{A}_\emptyset^2 + 
\sum_{\alpha \neq \emptyset}
\wh{A}_\alpha^2\E\left[\chi_\alpha(xx')\right]\right] \geq
\left(\E_v\left[\wh{A}_\emptyset\right]\right)^2 + \E_v\left[
\sum_{\alpha \neq \emptyset}
\wh{A}_\alpha^2\E\left[\chi_\alpha(xx')\right]\right]. \nonumber
\end{equation}
Lemma \ref{lem-4ss-oddeven} implies that
$\E\left[\chi_\alpha(xx')\right] \geq -\eps/2$. Using Parseval's we
obtain the lemma. Also, by symmetry $\E\left[A(x)A(x')\right] 
= \E\left[B(y)B(y')\right]$. \end{proof}

Observe that $x$ and $x'$ individually are independent of the pair
$(y,y')$. Similarly, $y$ and $y'$ individually are independent of the
pair $(x,x')$. Thus, we obtain,
\begin{align}
\E[A(x)B(y)B(y')] & = \E[A(x')B(y)B(y') \nonumber \\
		& = \E[A(x)A(x')B(y)] \nonumber \\
		& =  \E[A(x)A(x')B(y')] \nonumber \\
		& = \E[A(x)]\E[B(y)B(y')] \nonumber \\
		& \geq p_u(p_u^2 - \eps/2) \nonumber \\
		& \geq p_u^3 - \eps/2. \label{eqn-4ss-cube}
\end{align}
We are left with analyzing the term
$\E\left[A(x)A(x')B(y)B(y')\right]$. Fix the choice of $v$ and $w$ for
now. The Fourier expansion along with
standard arguments (analogous to those in earlier sections) yield,
\begin{align}
 \E\left[A(x)A(x')B(x)B(y')\right] = & \sum_{\alpha,
\beta}\wh{A}_\alpha^2\wh{B}_\beta^2\E\left[\chi_\alpha(xx') 
\chi_\beta(yy')\right] \nonumber \\
= & \sum_{\substack{\alpha, \beta \\
\pi_{vu}(\alpha)\cap\pi_{wu}(\beta) = \emptyset}} 
\wh{A}_\alpha^2\wh{B}_\beta^2 \E\left[\chi_\alpha(xx')\right] 
\E\left[\chi_\beta(yy')\right] \nonumber \\ 
& \ \ \ \ \ \ \ \ \ \ + 
\sum_{\substack{\alpha, \beta \\
\pi_{vu}(\alpha)\cap\pi_{wu}(\beta) \neq \emptyset}}\wh{A}_\alpha^2\wh{B}_\beta^2\E\left[\chi_\alpha(xx') 
\chi_\beta(yy')\right]. \label{eqn-4ss-split1}
\end{align}
It is easy to see that Lemma \ref{lem-4ss-oddeven} implies that
$\E\left[\chi_\alpha(xx')\right]\E\left[\chi_\beta(yy')\right] \geq
-\eps/2$. Thus, the first summation in the RHS of Equation
\eqref{eqn-4ss-split1} is at least $
\wh{A}_\emptyset^2\wh{B}_\emptyset^2 - (\eps/2)\left(\sum_\alpha
\wh{A}_\alpha^2\right)\left(\sum_\beta \wh{B}_\beta^2\right) =
\wh{A}_\emptyset^2\wh{B}_\emptyset^2
-\eps/2$. Using this we obtain,
\begin{equation}
\E\left[A(x)A(x')B(x)B(y')\right] \geq
\wh{A}_\emptyset^2\wh{B}_\emptyset^2 +
\sum_{\substack{\alpha, \beta \\
\pi_{vu}(\alpha)\cap\pi_{wu}(\beta) \neq \emptyset}}
\wh{A}_\alpha^2\wh{B}_\beta^2\E\left[\chi_\alpha(xx') 
\chi_\beta(yy')\right]\ \  - \frac{\eps}{2}. \nonumber
\end{equation}
Taking a further expectation over $v$ and $w$ and applying Jensen's
inequality 
we obtain,
\begin{equation}
\E_{\substack{v,w \\ x,x' \\ y,y'}}\left[A(x)A(x')B(x)B(y')\right] \geq 
p_u^4 + \E_{v,w}\left[\sum_{\substack{\alpha, \beta \\
\pi_{vu}(\alpha)\cap\pi_{wu}(\beta) \neq \emptyset}}
\wh{A}_\alpha^2\wh{B}_\beta^2\E\left[\chi_\alpha(xx') 
\chi_\beta(yy')\right]\right] - \frac{\eps}{2}. \label{eqn-4ss-4pw}
\end{equation}
Combining the above inequality with our assumption on the probability
of rejection of the verifier, along with Equations
\eqref{eqn-4ss-one},
\eqref{eqn-4ss-cross}, \eqref{eqn-4ss-cube}, and Lemma \ref{lem-4ss-xx}, 
yields,
\begin{align}
\rho^4 - \delta \geq & \frac{1}{16}\E_u\left[4p_u + 4p_u^2 + 2p_u^2 -
\eps + 4p_u^3 - 2\eps + p_u^4 - \frac{\eps}{2}\right]
\nonumber \\
 & + \frac{1}{16}\E_{u,v,w}\left[\sum_{\substack{\alpha, \beta \\
\pi_{vu}(\alpha)\cap\pi_{wu}(\beta) \neq \emptyset}}
\wh{A}_\alpha^2\wh{B}_\beta^2\E\left[\chi_\alpha(xx') 
\chi_\beta(yy')\right]\right] \nonumber \\
\geq & \frac{1}{16}\E_u\left[(1+p_u)^4 -4\eps\right] + 
 \frac{1}{16}\E_{u,v,w}\left[\sum_{\substack{\alpha, \beta \\
\pi_{vu}(\alpha)\cap\pi_{wu}(\beta) \neq \emptyset}}
\wh{A}_\alpha^2\wh{B}_\beta^2\E\left[\chi_\alpha(xx') 
\chi_\beta(yy')\right]\right] \nonumber \\
\geq & \frac{1}{16}\left[\left(1+\E_u[p_u]\right)^4 - 4\eps\right] + 
\frac{1}{16}\E_{u,v,w}\left[\sum_{\substack{\alpha, \beta \\
\pi_{vu}(\alpha)\cap\pi_{wu}(\beta) \neq \emptyset}}
\wh{A}_\alpha^2\wh{B}_\beta^2\E\left[\chi_\alpha(xx') 
\chi_\beta(yy')\right]\right],
\end{align}
where Jensen's inequality is used to obtain the last inequality.
Substituting the value $\rho$ from Equation \eqref{eqn-rho-frac} in the above and simplifying we obtain,
\begin{equation}
\E_{u,v,w}\left[\sum_{\substack{\alpha, \beta \\
\pi_{vu}(\alpha)\cap\pi_{wu}(\beta) \neq \emptyset}}
\wh{A}_\alpha^2\wh{B}_\beta^2\E\left[\chi_\alpha(xx') 
\chi_\beta(yy')\right]\right] \leq -16\delta + 4\eps,
\label{eqn-4ss-bd}
\end{equation}
where the inner expectation is over the choice of $x, x', y$ and $y'$. 
Before proceeding, we need the following lemma which follows from the
way $x, x', y, y'$ are chosen by the verifier.
\begin{lemma}\label{lem-4ss-JK}
For $i \in [k]$, let $J\subseteq \pi_{vu}^{-1}(i)$ and $K \subseteq
\pi_{wu}^{-1}(i)$ be non-empty subsets. Then,
\begin{equation}
\E\left[\chi_J(xx')\chi_K(yy')\right] = \begin{cases}
					(1-\eps) &\tn{ if both }|J|,
|K|\tn{ even,} \\
					-(1-\eps) &\tn{ if both }|J|,
|K|\tn{ odd,} \\
0 &\tn{otherwise.}
					\end{cases}
\nonumber
\end{equation}
\end{lemma}
Combining the above lemma with Lemma \ref{lem-4ss-J}, we obtain that,
\begin{align}
\left|\E\left[\chi_\alpha(xx')\chi_\beta(yy')\right]\right| \leq &
\left[\left(1 - \frac{\eps}{2}\right)^{\left|(\pi_{vu}(\alpha)\cup
\pi_{wu}(\beta)\setminus
(\pi_{vu}(\alpha)\cap\pi_{wu}(\beta))\right|}\right]\cdot\left[
\left(1-\eps\right)^{\left|\pi_{vu}(\alpha)\cap\pi_{wu}(\beta)\right|}\right]
\nonumber \\ \leq &  
\left(1 - \frac{\eps}{2}\right)^{\max\{\left|\pi_{vu}(\alpha)\right|,
\left|\pi_{wu}(\beta)\right|\}}.\label{eqn-4ss-albeta}
\end{align}
Let $R$ and $T$ ($R \geq T$) be positive integers we shall fix later. Using the
above we have,
\begin{align}
\sum_{\substack{\alpha, \beta \\
\pi_{vu}(\alpha)\cap\pi_{wu}(\beta) \neq \emptyset}}
\wh{A}_\alpha^2\wh{B}_\beta^2\left|\E\left[\chi_\alpha(xx') 
\chi_\beta(yy')\right]\right|
\leq & 
\sum_{\substack{|\alpha| < R,|\beta| < R \\ \pi_{vu}(\alpha)\cap\pi_{wu}(\beta)
\neq
\emptyset}}\wh{A}_{\alpha}^2\wh{B}_{\beta}^2
\nonumber \\  & +
\sum_{\substack{\big[ (|\alpha| \geq R, |\pi_{vu}(\alpha)| < T) \\ \vee 
(|\beta| \geq R, |\pi_{wu}(\beta)| < T)\big], \\ \pi_{vu}(\alpha)\cap\pi_{wu}(\beta)
\neq
\emptyset}}\wh{A}_{\alpha}^2\wh{B}_{\beta}^2 \nonumber \\
& + 
\sum_{\substack{\big[ (|\alpha| \geq R, |\pi_{vu}(\alpha)| \geq T) \\ \vee 
(|\beta| \geq R, |\pi_{wu}(\beta)| \geq T)\big], \\ \pi_{vu}(\alpha)\cap\pi_{wu}(\beta)
\neq
\emptyset}}\wh{A}_{\alpha}^2\wh{B}_{\beta}^2\left(1 -
\frac{\eps}{2}\right)^T
\end{align}
By Parseval's, the second term in the RHS above is at most,
$$\sum_{\substack{|\alpha| \geq R,\\ \pi_{vu}(\alpha) <
T}}\wh{A}_\alpha^2 + \sum_{\substack{|\beta| \geq R,\\ \pi_{wu}(\beta) <
T}}\wh{B}_\beta^2,$$
and the third term is at most,
$$\left(1 - \frac{\eps}{2}\right)^{T}.$$
We set $T = R^{c_0}$ where $c_0$ is the constant from Theorem
\ref{thm-LC}, and using the above analysis and Equation
\eqref{eqn-pilargebd}, we obtain,
\begin{align}
\left|\E_{u,v,w}\left[\sum_{\substack{\alpha, \beta \\
\pi_{vu}(\alpha)\cap\pi_{wu}(\beta) \neq \emptyset}}
\wh{A}_\alpha^2\wh{B}_\beta^2\E\left[\chi_\alpha(xx') 
\chi_\beta(yy')\right]\right]\right| \ \leq\  & 
\E_{u,v,w}\left[\sum_{\substack{|\alpha| < R,|\beta| < R 
\\ \pi_{vu}(\alpha)\cap\pi_{wu}(\beta)
\neq \emptyset}}\wh{A}_{\alpha}^2\wh{B}_{\beta}^2\right] \nonumber \\ & +
\frac{2}{R^{c_0}} \nonumber \\ & + \left(1 - \frac{\eps}{2}\right)^{R^{c_0}}.
\end{align}
Let us set $R =
\left(\frac{2}{\eps}\log\left(\frac{1}{\eps}\right)\right)^{\frac{1}{c_0}}$
and $\eps = \delta$. Using the above equation in conjunction with
Equation \eqref{eqn-4ss-bd} yields,
$$\E_{u,v,w}\left[\sum_{\substack{|\alpha| < R,|\beta| < R 
\\ \pi_{vu}(\alpha)\cap\pi_{wu}(\beta)
\neq \emptyset}}\wh{A}_{\alpha}^2\wh{B}_{\beta}^2\right] \geq
10\delta.$$
This yields a randomized labeling as follows: 
for every vertex $v \in V$, choose $\alpha
\subseteq [m]$ with probability $\wh{A^v}_\alpha^2$ and select a
random label from $\alpha$. For a vertex $u \in U$, choose a random
neighbor $w$ of $u$ and assign $u$ the label $\pi_{wu}(j_w)$ to $u$
where $j_w$ is the label assigned to $w$. The expected fraction of edges
of the $\mc{L}$ satisfied by this labeling is,
$$(10\delta)\left(\frac{1}{R}\right)^2 = \Omega(\delta^{c'}),$$
for some constant $c' > 0$ depending on $c_0$.

\subsubsection{Choice of parameters}
Analogous to previous sections, choosing $r = (\log \log N)/4$ in 
Theorem \ref{thm-LC} we get that the reduction to {\sc Max-E$4$-Set-Splitting} 
is of size $n = N^{O(r)}2^{2^{3r}} \leq
N^{O(\log\log N)}$. The soundness of $\mc{L}$ is $2^{-\Omega(\log\log
N)} = 2^{-\Omega(\log\log n)}$. Combining this with the above analysis in
the NO Case, choosing $\delta = \frac{1}{(\log n)^c}$ for some
positive constant $c$ (depending on $c_0$ and $\gamma_0$) we obtain a
contradiction to our assumption on the probability of rejection of the
verifier. 

Thus, in the NO Case, the verifier rejects with probability at least
$\rho^4 - \frac{1}{(\log n)^c}$. This completes the proof of Theorem
\ref{thm-main4}.

\bibliographystyle{alpha}
\bibliography{Refs-4pcp}

\end{document}